\documentclass[11pt,technote,a4paper,onecolumn,twoside]{IEEEtran} 

\usepackage{times}

\usepackage{moreverb}

\usepackage{amsmath, amssymb}

\usepackage{multicol}

\usepackage{epsfig}
\usepackage{cite}
\usepackage{url}

\usepackage{theorem}
\theorembodyfont{\normalfont}

\newtheorem{remark}{Remark}

\newtheorem{example}{Example}

\newtheorem{proposition}{Proposition}

\newtheorem{conjecture}{Conjecture}

\DeclareMathOperator*{\argmax}{arg\,max}

\makeatletter 
\allowdisplaybreaks[4]
\def\@oddfoot{\textbf\footnotesize \hfil\thepage}
\def\@evenfoot{\textbf\footnotesize \thepage \hfil}
\def\QEDclosed{\mbox{\rule[0pt]{1.3ex}{1.3ex}}} % for a filled box
\def\QED{\QEDclosed} % default to closed

\makeatother

\title{Data-Unit-Size Distribution Model with Retransmitted Packet Size Preservation Property
and Its Application to Goodput Analysis for Stop-and-Wait Protocol:
Case of Independent Packet Losses
}

\author{Takashi Ikegawa
	\thanks{T.~Ikegawa is with Waseda Research Institute for Science and Engineering,
	Waseda University and 
	Graduate School of Mathematical Sciences, the University of Tokyo, Japan
	(e-mail: ikegawa@aoni.waseda.jp or tikegawa@ms.u-tokyo.ac.jp).}
	}

\begin{document}

\maketitle

\begin{abstract}

This paper proposes a data-unit-size distribution model to represent
the retransmitted packet size preservation (RPSP) property
in a scenario where independently lost packets are retransmitted by a stop-and-wait protocol.
RPSP means that retransmitted packets with the same sequence number
are equal in size to the packet of the original transmission,
which is identical to the packet generated from a message
through the segmentation function,
namely, generated packet.
Furthermore, we derive goodput formula 
using an approach to derive the data-unit-size distribution.
We investigate the effect of RPSP on frame size distributions and goodput
in a simple case when no collision happens
over the bit-error prone wireless network equipped with IEEE 802.11 Distributed Coordination Function,
which is a typical example of the stop-and-wait protocol.
Numerical results show that the effect gets stronger
as bit error rate increases and
the maximum size of the generated packets is larger than the mean size
for large enough packet retry limits
because longer packets will be repeatedly corrupted and retransmitted more times
as a result of RPSP.

\end{abstract}


\section{Introduction}

Transfers of data units over communication networks suffer frequently from failure
due to various reasons including bit errors, congestion and collision.
To provide an error-free transmission service of messages,
i.e., data units generated by reliable applications,
a sender requires to implement one or more 
communication protocols
that include error recovery function.
The error recovery function allows the sender to retransmit lost packets.
For example,
distributed coordination function (DCF) for IEEE 802.11 wireless local area networks
specifies a stop-and-wait protocol (SWP) 
to realize the error-recovery function 
in a simple manner
\cite{IEEE07}.

The packets, i.e., SWP-layer data units,
that have been corrupted or lost within the networks will be transmitted
by the error-recovery function.
In general,
such retransmitted packets with the same sequence number ({\tt seqNum}) are equal in size
to the packet in the original transmission.
We call this property retransmitted packet size preservation: RPSP.

The packet retransmission probability will depend on the size of frames,
which are data units that contain the packet and are transferred over physical links.
Typical situations include the case when frames
are lost due to bit errors
because the frame corruption probability is approximately proportional
to the frame size.

In papers \cite{IKE04_RPSP_WiOpt, IKE05_RPSP_WiOpt, IKE05_RPSP_MSWiM},
the effect of RPSP on the mean frame size was discussed
for bit error prone networks.
These papers showed that the mean frame size with RPSP is larger 
than that without RPSP
as bit error rate increase if the packet size distribution has dispersion.
The reason for this is that longer frames will be repeated corrupted more time due to RPSP.

The frame sizes affect several quality of service (QoS) parameters 
(e.g., goodput) for applications.
Consequently, the effect of RPSP on QoS parameters will appear in some cases.
However, in previous work on QoS parameter analysis
over links with bit errors,
such as studies for IEEE 802.11 DCF goodput analysis 
including \cite{BIA00, CHA03, CHE11}, 
the effect of RPSP was ignored.
For example, frame sizes are assumed to be constant
although actual frames size distribution has dispersion (e.g.\cite{MOL00, NA06}). 
The purpose of this paper is to propose a data-unit-size distribution model with RPSP
to represent among the sizes of respective data units 
(i.e., messages, generated packets, transferred packets and frames)
and to derive the goodput formula 
using an approach to derive the data-unit-size distribution.

The rest of the paper is organized as follows.
In the next section, we describe the communication network model
underlying our study.
Section~\ref{sec: size distributions} derives the forms of size distributions
of generated packets, transferred packets and frames.
Section~\ref{sec: goodput}
derives the form of goodput
and applies the result to an IEEE 802.11 DCF wireless network.
Section~\ref{sec: results} investigates the effect of RPSP
on the frame size distribution and goodput for actual message-size distributions
Finally, Section~\ref{sec: conclusion} summarizes this paper and mentions future work. 

\section{Communication network model}

	\begin{figure}
	\centering
	\epsfig{file=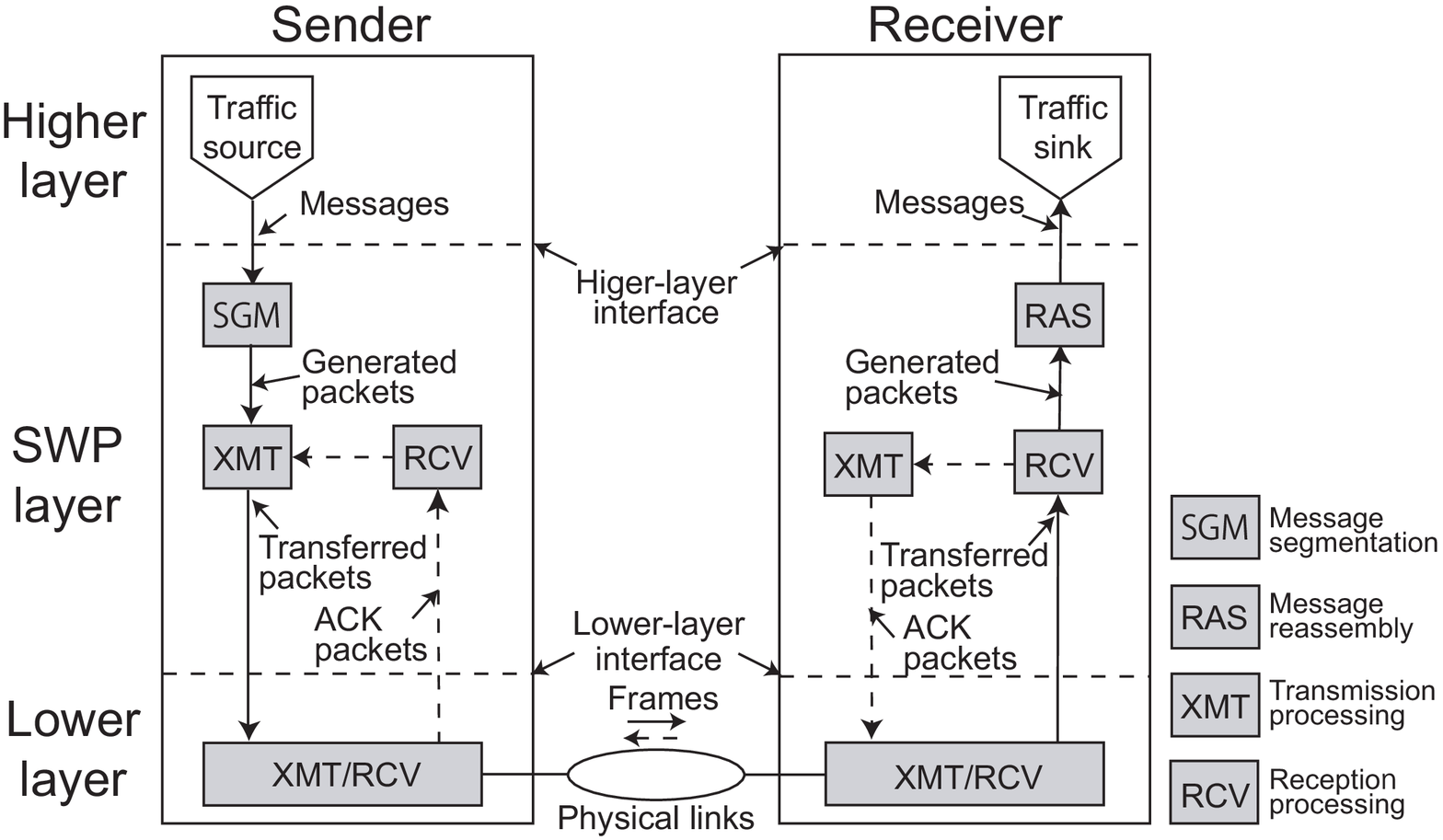, width=12cm} \\
	\caption{Communication network model.}
	\label{fig: network}
	\end{figure}

In this section,
we first explain the three-layered communication network model
under consideration.
Next,
the model of data units introduced in this paper
at the respective layer is described.
In final,
we explain some assumptions for analytical tractability.

\subsection{Layer model}

To characterize the nature of RPSP and message segmentation,
we consider a communication network
of which conceptual representation is shown in Fig.~\ref{fig: network}.
Each station (a sender and a receiver) has three layers.
The middle layer is referred to as an SWP layer.
It implements message segmentation-reassembly and error-recovery functions.
The error-recovery function is assumed to be implemented in a stop-and-wait scheme.
The layer above the SWP layer, namely the higher layer, 
contains a traffic source and sink.
The traffic source generates the data units.
On the other hand,
the traffic sink terminates the corresponding data units.
The layer below the SWP layer,
namely the lower layer,
contains an entity that can transfer data units over physical links at a sender.

\subsection{Data-unit model}

We define data units exchanged between peer entities at the respective layer as follows:

\begin{description}

\item{\bf Message:}
	a data unit generated by a traffic source
	with a given size distribution
	of which function is denoted by $F^{(m)}(\cdot)$.

\item{\bf Packet:}
	a data unit created from a message through segmentation function
	by adding a header and/or trailer, i.e., control information,
	to the (divided) message.
	We assume that size of SWP-layer's control information is constant
	and equal to $\ell_h^{({\rm R})}$.
	Whenever a packet is created,
	a {\tt seqNum} ($\ge 1$) is assigned.
	To model the RPSP explicitly,
	the packets are categorized into the following two kinds:

	\begin{description}

	\item{\bf Generated packet:}
		a packet that is generated from a message
		by a sender's SWP layer
		at the original transmission.
		The message segmentation function implemented in the sender's SWP layer
		enables a single message to be divided
		into several generated packets
		if the message size is larger than the
		payload size $\ell_d (>0)$.
		The receiver's SWP layer performs a message reassembly function,
		thus reassembling the segmented generated packets before delivering them
		to the higher layer.

	\item {\bf Transferred packet:}	
		a packet that is
		encapsulated into the frame.
		Due to RPSP,
		all the sizes of transferred packets
		with the same {\tt seqNum} 
		are equal to that of the generated packet.	
\end{description}

\item{\bf Frame:}
	a data unit that is made by encapsulating a transferred packet into a frame and
	by adding control information to the transferred packet,
	and will be transferred over physical links.
	The size of lower-layer's control information is assumed to be cosntant
	and equal to $\ell_h^{({\rm L})}$.

\end{description}

\subsection{Assumptions}

For analytical tractability,
we make the following assumptions.

\begin{description}

\item[\texttt{A1}:]
	Message sizes are mutually independent and 
	identically distributed 
	according to a common message-size distribution function $F^{(m)}(\cdot)$.
	The distribution $F^{(m)}(\cdot)$ has a finite mean value $\ell^{(m)}$,
	which is referred to as the mean message size.

\item[\texttt{A2}:]
	Frames are independently lost
	with probability
	\begin{align}
	g(x), \ 0 \le g(x) < 1, \notag
	\end{align}
	where $x$ is the size of information field in the frame,
	equivalently, the size of a transferred packet.

\item[\texttt{A3}:]
	The sender operates under a heavy traffic assumption,
	meaning that the sender's SWP layer always has a generated packet
	available to be sent.

\end{description}

\medskip

\begin{example}{\it Case of independent bit error prone links.}
Typical situations satisfying assumption \texttt{A2}
include the cases where frames are lost due to bit errors 
that occur independently.
Letting  $p_e$ be bi-error rate, $g(x)$ is given by
\begin{align}
\label{eq: bit}
		g(x) &= 1 - (1 - p_e)^{x+ \ell_h^{({\rm L})}},
\end{align}
where $x$ is the size of transferred packets.
\end{example}
%

%%%%%%%%%%%%%%%%%%%%%%%%%%%%%%%%%%%%%%%%%%%%%%%%%%%%%%%%

\section{Analysis of size distributions for generated packets,
	transferred packets and frames}
\label{sec: size distributions}

In this section,
we derive the forms of a size distributions of generated packets,
	transferred packets and frames under assumptions mentioned
	in the preceding section.

\subsection{Form of generated packet size distribution}

Let random variable $L^{(p)}$ be
	a size of generated packets.
Denoting $F^{(p)}(\cdot)$ be a generated packet size distribution,
that is $F^{(p)}(x) \stackrel{\triangle}{=}  \Pr.({L^{(p)} \le x})$,
from the argument \cite{IKE12_PerEva},
we have
\begin{align}
\label{eq: F p}
    F^{(p)}(x) 
	&= \left(1 - \pi^{(\textrm{E})}\right) \textbf{1}(x - \ell_d - \ell_h^{({\rm R})})
	+ \pi^{(\textrm{E})} F^{(\textrm{E})}(x),
\end{align}
where $\pi^{(\textrm{E})}$ is an occurrence probability of edge packets
and $F^{(\textrm{E})}(\cdot)$
is a distribution of edge-packet sizes.
The edge packet is defined as the final segmented generated-packet,
if a message is segmented.
It is identical with the original message if not segmented.

The forms of $\pi^{(\textrm{E})}$ and $F^{(\textrm{E})}(\cdot)$ are given by
\begin{align}
\label{eq: pi^E}
\pi^{(\textrm{E})} &= 
	\cfrac{1} 
	{
	\displaystyle\sum_{s=0}^\infty \, \int_{s \ell_d}^\infty \, dF^{(m)}(x)
	}
	=
	\cfrac{1} 
	{
	\displaystyle\sum_{s=0}^\infty \, \left(1 - F^{(m)}(s \ell_d)\right),
	}
\end{align}
and
\begin{align}
\label{eq: F^E}
F^{(\textrm{E})}(x) 
&= \begin{cases}
	0, 	& 0 \le x < \ell_h^{({\rm R})}, \\
	\displaystyle \sum_{s = 0}^{\infty}
	\left\{F^{(m)}(x + s \, \ell_d - \ell_h^{({\rm R})}) - F^{(m)}(s \, \ell_d)\right\},
		& \ell_h^{({\rm R})} \le x \le \ell_d + \ell_h^{({\rm R})}, \\
	1, 	& x > \ell_d + \ell_h^{({\rm R})}.
\end{cases}
\end{align}

\begin{example}{\it Case of discrete message-size distribution.} 
\label{example: Case of discrete message-size distributions}
Consider the case where the message-size distribution function $F^{(m)}(\cdot)$
is given by 
\begin{align}
\label{eq: F m_d}
F^{(m)}(x) &=\sum_{i=1}^{n_d} \omega^{(m)}_i \textbf{1}(x - \ell^{(m)}_i),
\end{align}
where $n_d \ge 1$, $w^{(m)}_i > 0$, $\ell^{(m)}_i > 0$
	for $i = 1, 2, \cdots, n_d$, and $\sum_{i=1}^{n_d} w^{(m)}_i = 1$.

The form of $\pi^{(\textrm{E})}$ is given by $\{\sum_{i=1}^{n_d} w^{(m)}_i k_i\}^{-1}$ 
with $k_i = \lceil \ell^{(m)}_i / \ell_d \rceil$.
This can be intuitively shown from the fact
that 1) $k_i$ generated packets are created from one message of size $\ell^{(m)}_i$, and 
2) they consist of $k_i - 1$ generated packets of size $\ell_d$
(called body packets \cite{IKE12_PerEva}) and one edge packet.
The generated-packet-size distribution can be written as
\begin{align}
\label{eq: F p d}
	F^{(p)}(x) &= 
	\left(1 - \pi^{(\textrm{E})}\right) \, 
	\textbf{1}(x - \ell_d - \ell_h^{({\rm R})})
	+ \pi^{(\textrm{E})} \, \sum_{i=1}^{n_d} \,
	w^{(m)}_i \, 
	\textbf{1}(x - \ell^{(m)}_i + (k_i - 1) \, \ell_d - \ell_h^{({\rm R})}).
\end{align}
The form of \eqref{eq: F p d} can be rewritten as
\begin{align}
\label{eq: F p d another 1}
	F^{(p)}(x) &\stackrel{\triangle}{=}  
	\sum_{i=0}^{n_d} w^{(p)}_i \textbf{1}(x - \ell^{(p)}_i),
\end{align}
where
\begin{align}
\label{eq: F p d another 2}
&\begin{cases}
	w^{(p)}_0 = 1 - \pi^{(\textrm{E})},\\
	l^{(p)}_0 = \ell_d + \ell_h^{({\rm R})},
\end{cases} \\
&\begin{cases}
	w^{(p)}_i = \pi^{(\textrm{E})} \, w^{(m)}_i,\\
\label{eq: F p d another 3}
	l^{(p)}_i = \ell^{(m)}_i - (k_i - 1) \, \ell_d + \ell_h^{({\rm R})},
\end{cases}i=1, 2, \cdots, n_d.
\end{align}
\hspace*{\fill}~\QED
\end{example}

Letting $\ell^{(p)}$ be the mean packet size,
we have
\begin{align}
\label{eq: ell p}
\ell^{(p)} &\stackrel{\triangle}{=} \int_{0}^\infty \,x \,dF^{(p)}(x) 
= \pi^{(\textrm{E})} \, \ell^{(m)} + \ell_h^{({\rm R})}.
\end{align}

\subsection{Form of transferred packet size distribution}

Let $F^{(q)}(\cdot)$ be a transferred packet size distribution.
Denoting the number of retransmissions of the transferred packet
with the same {\tt seqNum}
of the generated packet
which size is equal to $L^{(p)}$
by $R$,
we can prove the following proposition.

\begin{proposition}
\label{pro: F q}
The transferred packet size distribution $F^{(q)}(\cdot)$
is given by
\begin{align}
\label{eq: F q}
F^{(q)}(y)
	&= \begin{cases}
	0, 	&	0 \le y < \ell_h^{({\rm R})}, \\
	\cfrac{
		\displaystyle \int_{x=\ell_h^{({\rm R})}}^{x=y} 
		E\left[R + 1 \, | \,  L^{(p)} = x \right] \, dF^{(p)}(x)
		}
		{
		E\left[R + 1 \right]
		}, 	& 
			\ell_h^{({\rm R})} \le y \le \ell_d + \ell_h^{({\rm R})}, \\
		1,	&	
			y > \ell_d + \ell_h^{({\rm R})}.
	\end{cases}
\end{align}
\end{proposition}

\begin{proof}
See Appendix~\ref{ap: F q}.
\end{proof}

\medskip

From assumption \texttt{A2},
the form of $E[R + 1 \, | \,  L^{(p)} =x]$
for $\ell_h^{({\rm R})} \le x \le \ell_d + \ell_h^{({\rm R})}$ is given by
\begin{align}
	E\left[R + 1 \, | \,  L^{(p)} =x\right] &= 
	 \left(1-g(x)\right) \, \sum_{r=0}^{n_\textrm{RL}} \, (r+1) \, \left\{g(x)\right\}^r  
	+ \left\{g(x)\right\}^{n_\textrm{RL}+1} \, \left(n_\textrm{RL} + 1\right)  \notag \\
	&= \cfrac{1-\{g(x)\}^{n_\textrm{RL}+1}}{1 - g(x)}
	\stackrel{\triangle}{=} h(x, n_\textrm{RL}).
\end{align}
where $n_{\rm RL}(\ge 0)$ is the maximum number of retransmission attempts of the transferred packet
	with the same {\tt seqNum},
	referred to as retry limit.

\medskip

\begin{remark}
Substitution of \eqref{eq: F p} into \eqref{eq: F q} yields $F^{(q)}(y)$
for $\ell_h^{({\rm R})} \le y \le \ell_d + \ell_h^{({\rm R})}$ given by
\begin{align}
	F^{(q)}(y) &=
		\cfrac{
		\left(1 - \pi^{(\textrm{E})}\right) 
		\, h(\ell_d + \ell_h^{({\rm R})}, n_\textrm{RL}) \, \textbf{1}(x - \ell_d - \ell_h^{({\rm R})})
		+ \pi^{(\textrm{E})} \, \displaystyle\int_{x=\ell_h^{({\rm R})}}^{x=y} 
			\, h(x, n_\textrm{RL}) \, dF^{(\textrm{E})}(x)
		}
		{
		E\left[R + 1 \right]
		},	
\end{align}
where $E[R + 1]$ is given by
\begin{align}
	E\left[R + 1 \right] &=
	\left(1 - \pi^{(\textrm{E})}\right) \, h(\ell_d + \ell_h^{({\rm R})}, n_\textrm{RL})
	+
	\pi^{(\textrm{E})} \,
	\int_{\ell_h^{({\rm R})}}^{\ell_d + \ell_h^{({\rm R})}} 
			\, h(x, n_\textrm{RL})
			\, dF^{(\textrm{E})}(x).
\end{align}
\hspace*{\fill}~\QED 
\end{remark}

\medskip

\begin{example} 
\label{ex: no frame loss}
\textit{RPSP effect when no frame is lost.}
Consider the case where no frame is lost.
In this case, the number of retransmissions is equal to zero, i.e., $R = 0$.
From \eqref{eq: F q}, $F^{(q)}(x)$ is identified with $F^{(p)}(x)$,
implying that no effect of RPSP appears. 
\hspace*{\fill}~\QED
\end{example}

\medskip

\begin{example} 
\label{ex: common generated packet size}
\textit{RPSP effect when generated packets are constant in size.}
Let us consider the case where generated packets have a common size
$\ell_c (= \ell^{(p)})$,
that is
\begin{align}
\label{eq: F p c}
	F^{(p)}(x) &= \textbf{1}(x - \ell_c).
\end{align}
Thypical situations include 
when message sizes follow the discrete distribution function
given by \eqref{eq: F m_d} with $n_d = 1$ and $\ell^{(m)}_1 (= \ell_c - \ell_h^{({\rm R})}) \le \ell_d$,
resulting in $\pi^{(\textrm{E})} = 0$.
Note that $F^{(p)}(x)$ can be approximated by $\textbf{1}(x - \ell_d)$
if $\ell^{(m)}$ is large enough compared with $\ell_d$
from \cite[Remark 3]{IKE12_PerEva}.

With \eqref{eq: F q} and \eqref{eq: F p c},
$F^{(q)}(x)$ is identified with $F^{(p)}(x) = \textbf{1}(x - \ell_c)$,
which indicates that no effect of RPSP appears. 
\hspace*{\fill}~\QED
\end{example}

\subsection{Form of frame size distribution}

Denote the frame size distribution by $F^{(f)}(\cdot)$.
Since a frame contains a transferred packet and
the size of control information added the transferred packet
is $\ell_h^{({\rm L})}$,
$F^{(f)}(x)$ is simply given by $F^{(q)}(x - \ell_h^{({\rm L})})$.

%%%%%%%%%%%%%%%%%%%%%%%%%%%%%%%%%%%%%%%%%%%%

\section{Goodput Analysis}
\label{sec: goodput}

In this section,
first,
we derive the form of goodput in a simple scenario.
Next, we apply the result to an IEEE 802.11 DCF wireless network.

\subsection{Form of goodput}

Let $G$ be goodput of a single SWP connection,
which is defined as the mean number of bits 
by a receiver's higher layer entity
across the higher layer interface
per unit time.
We denote  the interdeparture time of the transferred packet
by $T^{(cycle)}$.
In addition,
we denote the event meaning
that the transferred packet is successfully transmitted
by ``{\it delivery}''.
Then we can prove the following proposition.

\begin{proposition}
\label{pro: G}

The form of goodput $G$ is given by
\begin{align}
\label{eq: G}
	G &= \cfrac{
		\displaystyle\int_{x=\ell_h^{({\rm R})}}^{\ell_d + \ell_h^{({\rm R})}} \,  
			\Pr.(\text{\it delivery} \, | \, L^{(p)} = x) \, (x - \ell_h^{({\rm R})}) \, dF^{(p)}(x)
		}
		{
		\displaystyle\int_{x=\ell_h^{({\rm R})}}^{\ell_d + \ell_h^{({\rm R})}} \, 
			E\left[R + 1 \, | \,  L^{(p)} = x \right]
			\,E\left[T^{(cycle)} \, | \,  L^{(p)} =x \right] \, dF^{(p)}(x)
		}.
\end{align}

\end{proposition}

\begin{proof}
See Appendix~\ref{ap: goodput}.
\end{proof}

\medskip

Note that assumption {\tt A2} yields the form of
$\Pr.(\text{\it delivery} \, | \, L^{(p)} = x)$ given by
\begin{align}
	\Pr.(\text{\it delivery} \, | \, L^{(p)} = x)
	&= 1 - \{g(x)\}^{n_\textrm{RL}+1}.
\end{align}

\subsection{Application of goodput analysis to IEEE 802.11 DCF}

We consider a simple scenario where
just one sender and one receiver exist
in a wireless network equipped with IEEE 802.11 DCF,
which is an SWP protocol.
Since no collision occurs,
from the argument described in \cite{CHA03},
the form of $E[T^{(cycle)} \, | \, L^{(p)} = x]$ in \eqref{eq: G}
can be simply written as
\begin{align}
\label{eq: E T cycle L}
    E\left[T^{(cycle)} \, | \, L^{(p)} = x \right] &= 
	\cfrac{
	(1-g(x)) \, \sigma
	}
	{
	1-\left\{g(x)\right\}^{n_\text{RL}+1}
	}
	\displaystyle\sum_{r=0}^{n_\text{RL}}
	b_r \, \left\{g(x)\right\}^r \notag \\
	&\quad+ \left(1 - g(x)\right) \, t_{\textrm{suc}}(x) 
	+ g(x) \, t_{\textrm{bit}}(x),
\end{align}
where 
\begin{description}

\item[$\sigma$: ]\rule{0in}{0in}
	DCF backoff slot size

\item[$b_r$:]\rule{0in}{0in}
	mean value of the backoff counter of the $r$th backoff stage,
	i.e., the $r$th retransmission attempt of the transferred packet

\item[$t_\textrm{suc}(x)$ and $t_\textrm{bit}(x)$: ]\rule{0in}{0in}

	mean interdeparture times of the transferred packet
	of size of $x$
	when a transmission is successful and
	fails due to bit errors, respectively.

\end{description}

The value of $E[b_r]$ is equal to $\textrm{CW}_r / 2$
because the backoff time at each transmission is uniformly chosen
in the range $[0, \textrm{CW}_r]$
where $\textrm{CW}_r$ is 
$\min\{2^r \, (\textrm{CW}_{\min} + 1) - 1, \textrm{CW}_{\max}\}$
for $r = 1, 2, \cdots, n_\textrm{RL}$ and
$\textrm{CW}_0$ is $\textrm{CW}_{\min}$.
Assuming that propagation delay is negligible,
we have 
\begin{align}
	t_\textrm{suc}(x) &=
	\cfrac{x + \ell_{\rm ACK}}{\mu_d}
	+ \cfrac{2 \, \ell_h^{({\rm L})}}{\mu_b}
	+ t_{\rm SIFS}
	+ t_{\rm DIFS},\\
    	t_\textrm{bit}(x) &= \cfrac{x}{\mu_d}
	+ \cfrac{\ell_h^{({\rm L})}}{\mu_b}
	+ t_{\rm EIFS},
\end{align}
where
$\mu_d$ is data-transmission rate,
$\mu_b$ is basic-link rate,
and 
$\ell_{\rm ACK}$ is ACK-packet size.
Here,
$t_{\rm SIFS}$,  $t_{\rm DIFS}$ and $t_{\rm EIFS}$ are
Short Inter Frame Space (IFS), 
DCF IFS
and
Extended IFS, respectively.
The derivation of \eqref{eq: E T cycle L} can be found in
Appendix~\ref{ap: E T cycle}.

\section{Numerical results and discussions}
\label{sec: results}

In this section,
we examine the effect of RPSP on frame-size distributions and goodput
by utilizing the results in Sections~\ref{sec: size distributions} and \ref{sec: goodput}.
We consider a scenario in which 
Web objects are transferred over the IEEE~802.11 DCF network
where bit errors occur independently.
In the following, 
we used the parameter values listed in Table~\ref{ta: parameter}.

\begin{table}[htbp]
%\begin{minipage}{10cm}
\begin{center}
\caption{Parameter values used to obtain numerical results}
\label{ta: parameter}
\renewcommand\thefootnote{\alph{footnote}}
\begin{tabular}{|l||l|} \hline
Parameter			& Value	\\\hline\hline
Basic-link rate  $\mu_b$ 	& $1$~Mbps \\
Data-transmission rate  $\mu_d$ & $11$~Mbps \\
SWP layer information field size $\ell_h^{({\textrm{R}})}$
				& $34$~bytes \\
Lower layer information field size $\ell_h^{({\textrm{L}})}$	& $24$~bytes	\\
Slot time $\sigma$ 		& $20~\mu$sec \\ 
Short IFS $t_{\rm SIFS}$ 	& $10~\mu$sec \\
DCF IFS $t_{\rm DIFS}$ 		& $50~\mu$sec \\
%
%ACK timeout value $t_\textrm{TO}$ & $233~\mu$sec\footnotemark[1] \\
Extended IFS $t_{\rm EIFS}$ &  $263~\mu$sec \\
ACK-packet size $\ell_{\rm ACK}$ & $14$~bytes \\ 
Minimum contention window size $\textrm{CW}_{\min}$ & $31$\\
Maximum contention window size $\textrm{CW}_{\max}$ & $1023$ \\
\hline
\end{tabular}
\end{center}
%
%\footnotetext[1]{\scriptsize
%		The value of $t_\textrm{TO}$ is set to $t_{\textrm{SIFS}}
%		+ \ell_{\rm ACK}/\mu_d + \ell_h^{\rm PHY}/\mu_b + \sigma$.}
%
%\footnotetext[2]{\scriptsize
%	The value of $t_{\rm EIFS}$ is set to $t_\textrm{TO}+t_{\textrm{DIFS}}-\sigma$.}
%\end{minipage}
\end{table}

Two kinds of Web pages are considered: static and dynamic Web pages.
We shall use the following Web object size distributions
from traffic measurements \cite{MOL00,SHI03}.
\begin{itemize}
\item{{\bf Static Web objects:}} 
	The sizes of the static Web objects is assumed to follow a lognormal distribution given by
	\begin{align}
	\label{eq: F^{(m)}^l(x)}
	F^{(m)}(x) &= 
		\begin{cases}
		\displaystyle \int_{y=0}^{y=x} \cfrac{1}{\sqrt{2 \pi} \sigma y}
		e^{\tfrac{-\left(\log y - \mu\right)^2}
		{2 \sigma^2}} dy, & x > 0, \\
		0, & x \le 0.
		\end{cases}
	\end{align}

The distribution parameters $\mu$ and $\sigma$ are assumed to be $6.34$ and $2.07$, respectively,
on the basis of the measured mean message size $\ell^{(m)}=4827$ bytes and
the measured standard deviation $\sigma^{(m)}=41,008$~bytes.
Note that this lognormal distribution can represent a long-tailed property.

\item{\textbf{Dynamic Web objects:}}
	The sizes of the dynamic Web objects are assumed to follow a Weibull distribution:
	\begin{align}
	F^{(m)}(x) &= 
		 \begin{cases}
		1 - e^{-(\lambda \, x)^\nu}, & x > 0, \\
		0, & \text{$x \le 0$}.
		\end{cases}	
	\end{align}
	The scale parameter $\lambda$ and the shape parameter $\nu$ are assumed to be 
	$4.02 \times 10^{-4}$ and $1.9$, respectively,
	which fit the measured dynamic Web object size distribution for one case of an entertainment site \cite{SHI03}.  
	Note that the Weibull distribution in this case is not a long-tailed distribution
	because the shape parameter $\nu$ is not smaller than $1$.
	The mean message size $\ell^{(m)}$ is $2207.37$~bytes,
	and the standard deviation $\sigma^{(m)}$ is $1208.43$~bytes.
\end{itemize}

\subsection{Effect of RPSP on frame size distribution}

\begin{figure}
\begin{minipage}[t]{8cm}
\centering
	\epsfig{file=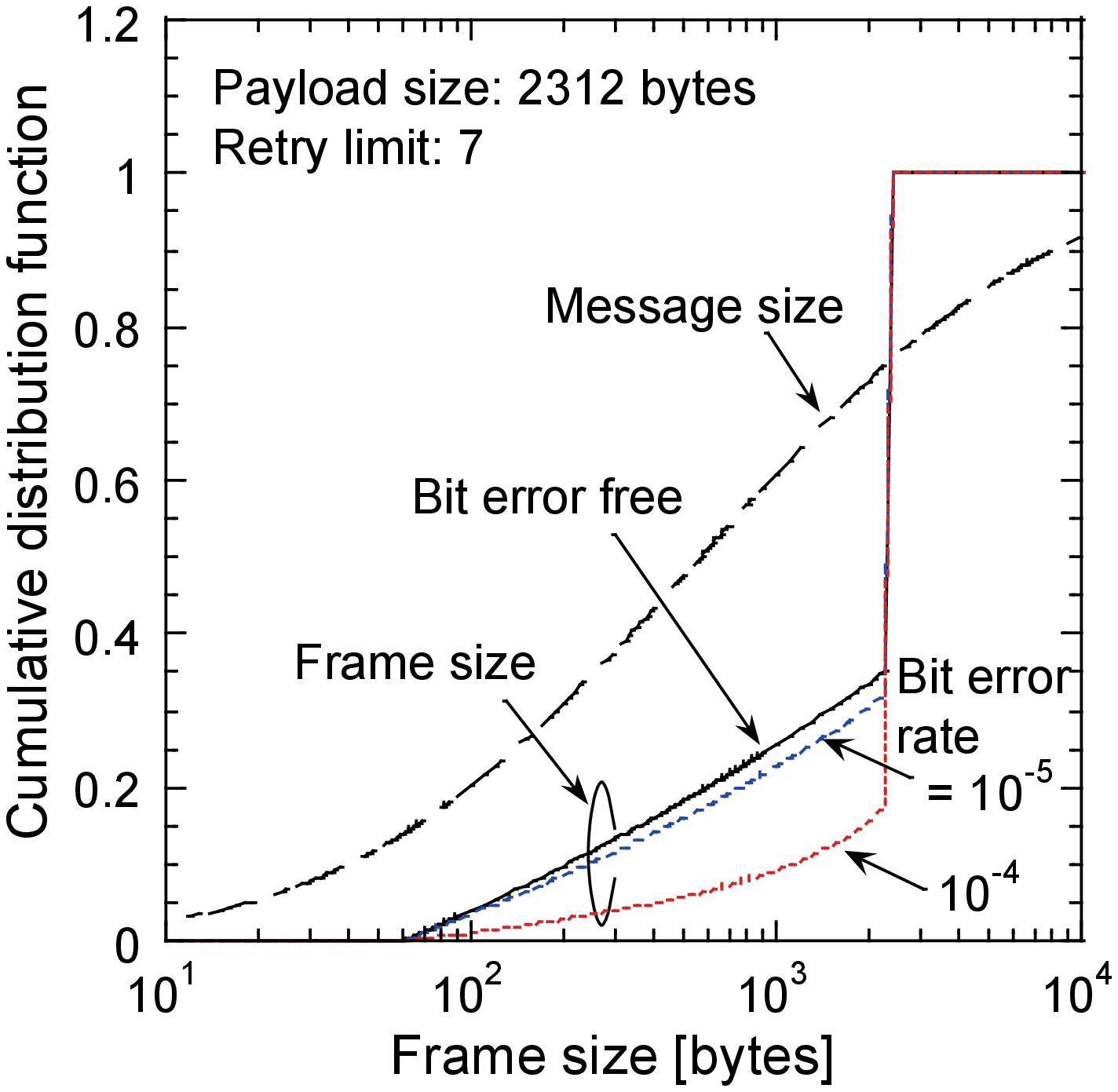, width=6cm} \\
	\footnotesize{a) Case of static Web objects.} 
\end{minipage}
\hfill
\begin{minipage}[t]{8cm}
\centering
	\epsfig{file=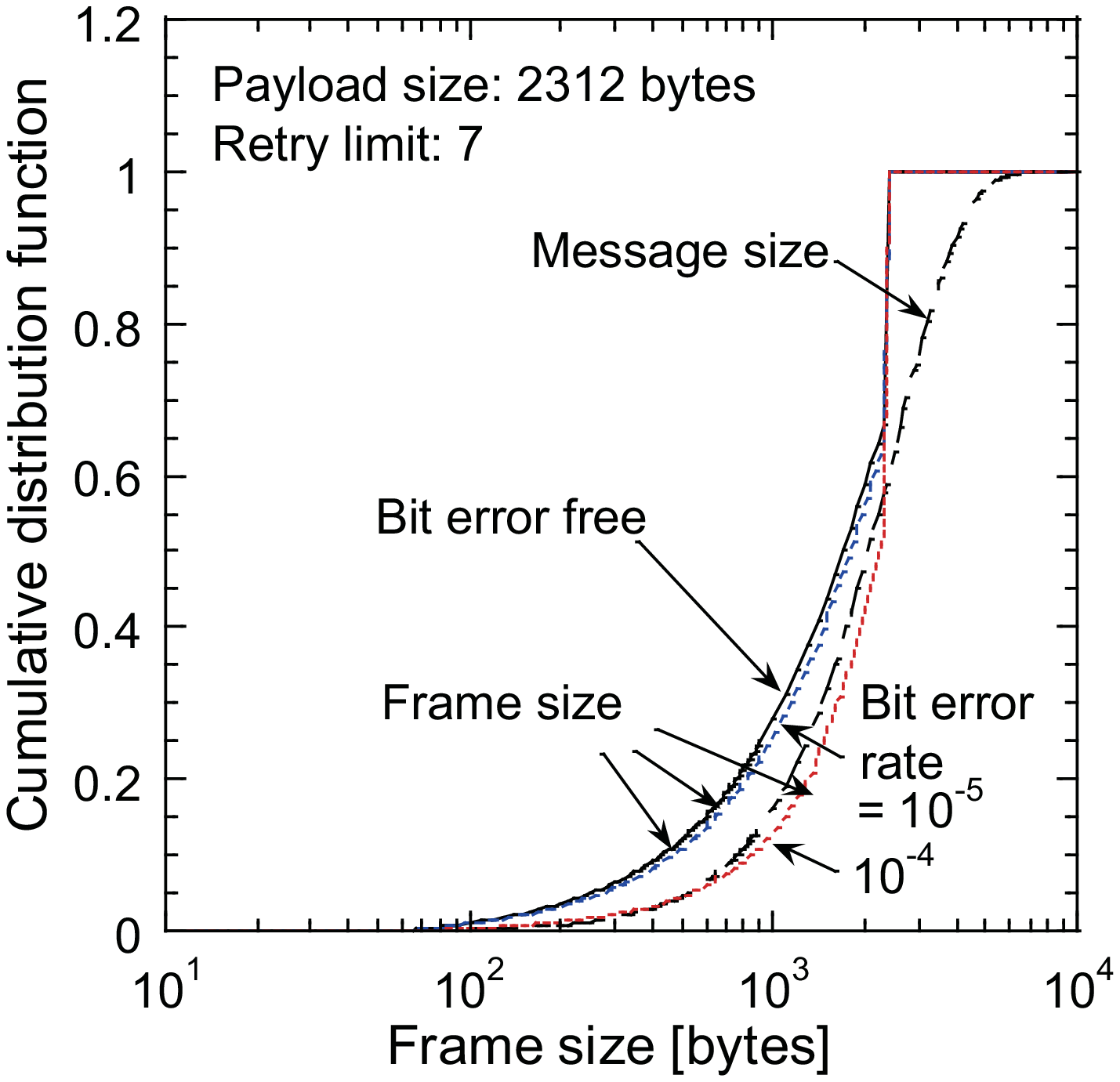, width=6cm} \\
	\footnotesize{b) Case of dynamic Web objects.}
\end{minipage}\\
\caption{Cumulative frame size distributions $F^{(f)}(\cdot)$
	for different bit error rates $p_e$.}
\label{fig: frame size}
\end{figure}

Figures~\ref{fig: frame size} (a) and (b) show the distributions of frame sizes
$F^{(f)}(\cdot)$ for different bit error rates $p_e$
of static and dynamic Web objects, respectively.
We used payload size $\ell_d$ of $2312$~bytes and
retry limit $n_\textrm{RL}$ of $7$.
Note that  $2312$~bytes of the payload size $\ell_d$ is the 
maximum transmission unit size of IEEE 802.11 wireless LANs
and
$7$ of retry limit $n_\textrm{RL}$ is the default value \cite{IEEE07}.

These figures show that
the frame size distribution $F^{(f)}(\cdot)$
for high bit error rates
is significantly different
from that for bit error free.
Thus,
we can see that the effect of RPSP produces a more concave curve
for the transferred packet size distribution
when the bit error rate is higher.

\begin{figure}
\begin{minipage}[t]{8cm}
\centering
	\epsfig{file=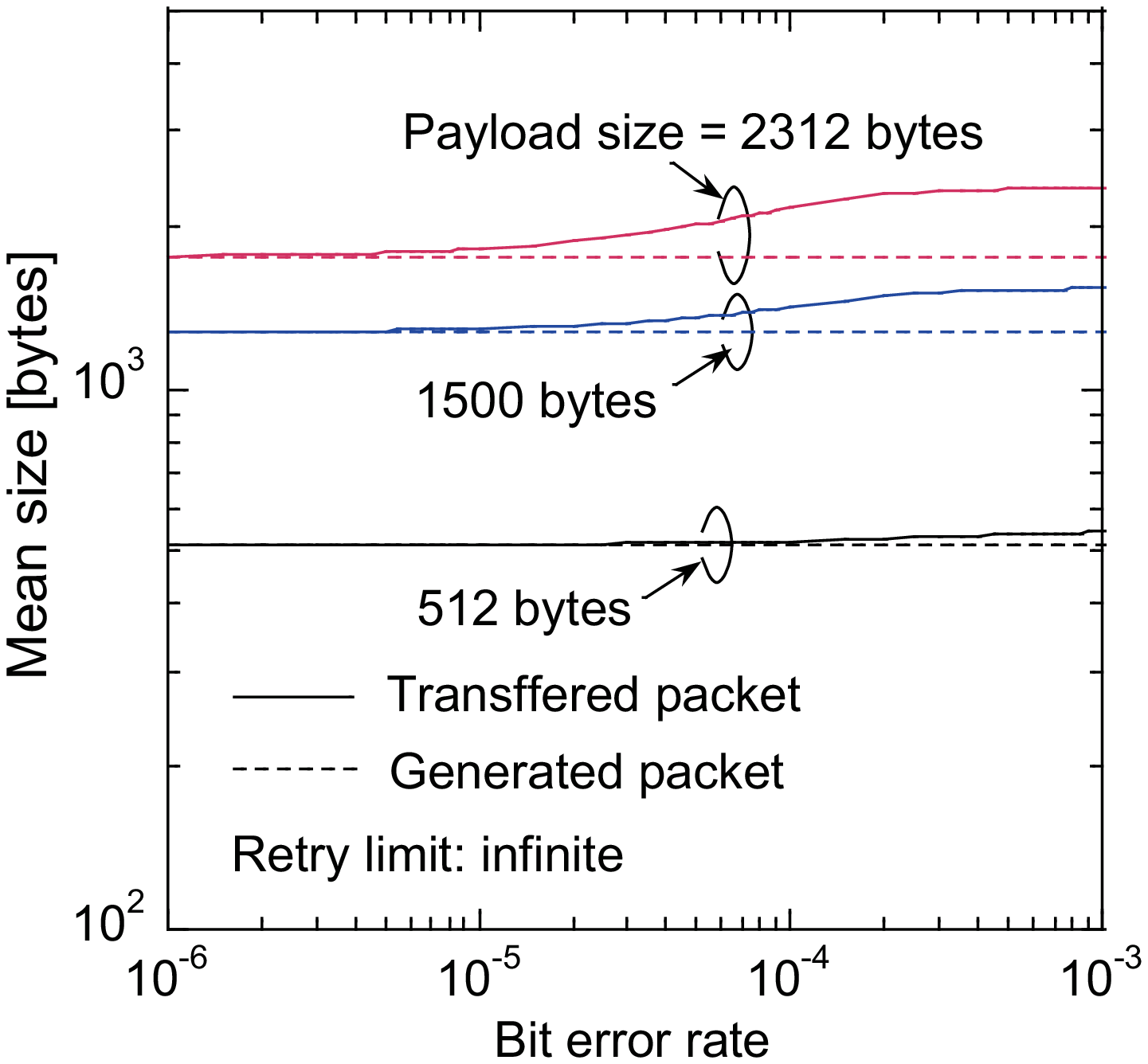, width=6cm} \\
	\footnotesize{a) Case of static Web objects.} 
\end{minipage}
\hfill
\begin{minipage}[t]{8cm}
\centering
	\epsfig{file=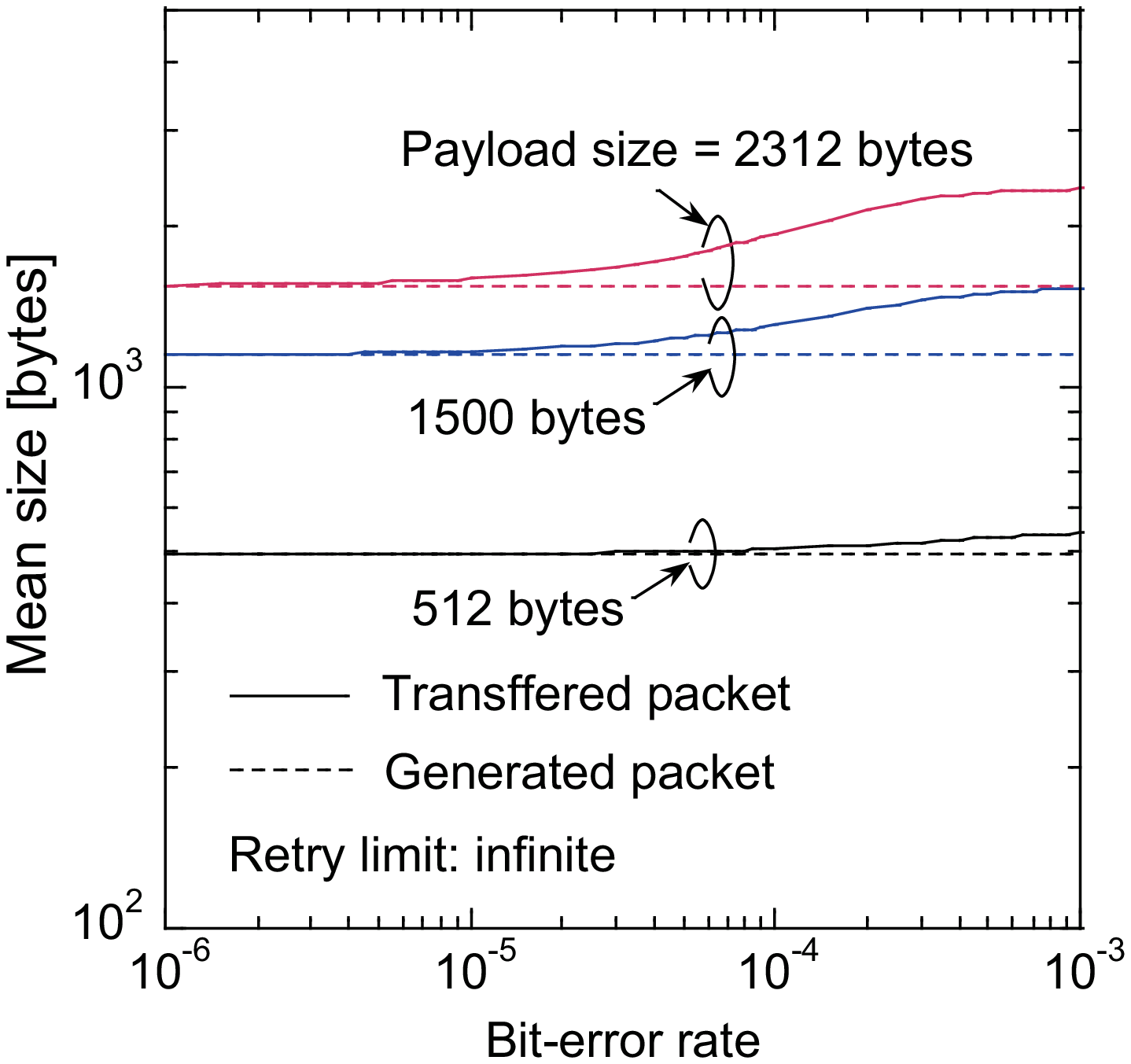, width=6cm} \\
	\footnotesize{b) Case of dynamic Web objects.}
\end{minipage}\\
\caption{Mean transferred packet size $\ell^{(q)}$ and 
mean generated packet size $\ell^{(p)}$ versus bit error rates $p_e$
when retry limit $n_\textrm{RL}$ goes to infinite
for different payload sizes $\ell_d$.}
\label{fig: l q l p}
\end{figure}

%%%%%%%%%%%%%%%%%%%%%%%%%

\begin{table}
\caption{Mean size of transferred packets $\ell^{(q)}$ for different bit error rates $p_e$
	when payload size $\ell_d$ is 2312~byte and
	retry limit $n_\textrm{RL}$ goes to infinite.}
\label{ta: ell q ell p}
\centering
\begin{tabular}{|l|| r | r | r | r |} 
\hline
\bfseries $p_e$ & \bfseries $10^{-6}$ & \bfseries $10^{-5}$ & \bfseries $10^{-4}$ & \bfseries $10^{-3}$
	\\ \hline\hline
Static Web objects			& $1761.4$ & $1815.0$ & $2161.4$ & $2344.6$ \\ \hline
Dynamic Web objects			& $1552.0$ & $1592.9$ & $1926.8$ & $2334.8$ \\ \hline
\end{tabular} \\
\small{
{\bfseries Note:}~Mean sizes of transferred packets $\ell^{(q)}$ are represented in units of bytes.
	Maximum size of generated packets $\ell^{(p)}_{\max}$ of static and dynamic Web objects
	is $2346.0$~bytes,
	which is $\ell_d + \ell_h^{({\rm R})}$.
}
\end{table}

Let $\ell^{(q)}$ be the mean transferred packet size,
that is $\ell^{(q)} \stackrel{\triangle}{=} \int_0^\infty \, x \, dF^{(q)}(x)$.
To investigate the effect of RPSP when retry limit $n_\textrm{RL}$ goes to infinite, 
Figs.~\ref{fig: l q l p} (a) and (b) show
mean transferred packet size $\ell^{(q)}$ and 
mean generated packet size $\ell^{(p)}$ of static and dynamic Web objects, respectively,
versus bit error rates $p_e$
for different payload sizes $\ell_d$.
Table~\ref{ta: ell q ell p} lists mean size of transferred packets $\ell^{(q)}$
for different bit error rates $p_e$
when payload size $\ell_d$ is 2312~byte and
retry limit $n_\textrm{RL}$ goes to infinite.
in the cases of static and dynamic Web objects.
From Figs.~\ref{fig: l q l p} (a) and (b), and Table~\ref{ta: ell q ell p},
we find that the RPSP effect appears
when the bit error rate $p_e$ exceeds $10^{-5}$.
The reason for this is that longer transferred packets are likely to be retransmitted more times.
Letting random variables $L^{(p)}_\kappa$ and $R_\kappa$ be 
size and the number of retransmissions of the transferred packet
of which {\tt seqNum} is $\kappa$, respectively,
this implies that
$h(x^{(p)}_\kappa, \infty) = E[R_\kappa + 1\,|\, L^{(p)}_\kappa = x^{(p)}_\kappa]
> h(x^{(p)}_{\kappa'}, \infty)$ if $x^{(p)}_\kappa > x^{(p)}_{\kappa'}$.

Let $\ell^{(p)}_{\max}$ be the maximum generated packet size,
i.e., $\ell^{(p)}_{\max} = \min\{l; F^{(p)}(l)=1\}$.\footnote{
	Letting $\ell^{(m)}_{\max}$ be the maximum message size,
	$\ell^{(p)}_{\max}$ is given by $\min\{\ell_d, \ell^{(m)}_{\max}\}$.
}
From an inspection of Figs.~\ref{fig: l q l p} (a) and (b), and Table~\ref{ta: ell q ell p},
we find that $\ell^{(q)}$ reaches around $\ell^{(p)}_{\max}$
as $p_e \to 1$.
This implies that
the number of transmissions of the longest transferred packets
is dominant in the total number of transmissions of all transferred packets
due to RPSP.
Then, we have the following conjecture.

\medskip

\begin{conjecture}
\label{conjecture: ell q max}
\textit{Asymptotic bound on mean transferred packet size.}
We denote 
the asymptotic bound on the mean transferred packet size by $\ell^{(q)}_{\max}$.
That is
the finite limit of the mean transferred packet size
as the value of $p_e$ approaches one.
Then, we have 
\begin{align}
\label{eq: ell q max}
\ell^{(q)} \to \ell^{(q)}_{\max} &= 
	\ell^{(p)}_{\max},
	\qquad\text{as $p_e \to 1$.}
\end{align}
\hspace*{\fill}~\QED
\end{conjecture}

\bigskip

Appendix~\ref{ap: w q infty} provides the proof of conjecture~\ref{conjecture: ell q max}
in the case of a discrete generated packet size distribution.

From conjecture~\ref{conjecture: ell q max},
we find that RPSP effect appears stronger 
when $\ell^{(p)}_{\max} / \ell^{(p)}$ increases. 
If the mean message size $\ell^{(m)}$ is enough large
compared with payload size $\ell_d$,
resulting in $\ell^{(p)} \approx \ell_d = \ell^{(p)}_{\max}$,
RPSP effect is likely to disappear.

%%%%%%%%%%% グッドプット %%%%%%%%%%%%%%%%%

\subsection{Effect of RPSP on goodput}

\begin{figure}
\begin{minipage}[t]{8cm}
\centering
	\epsfig{file=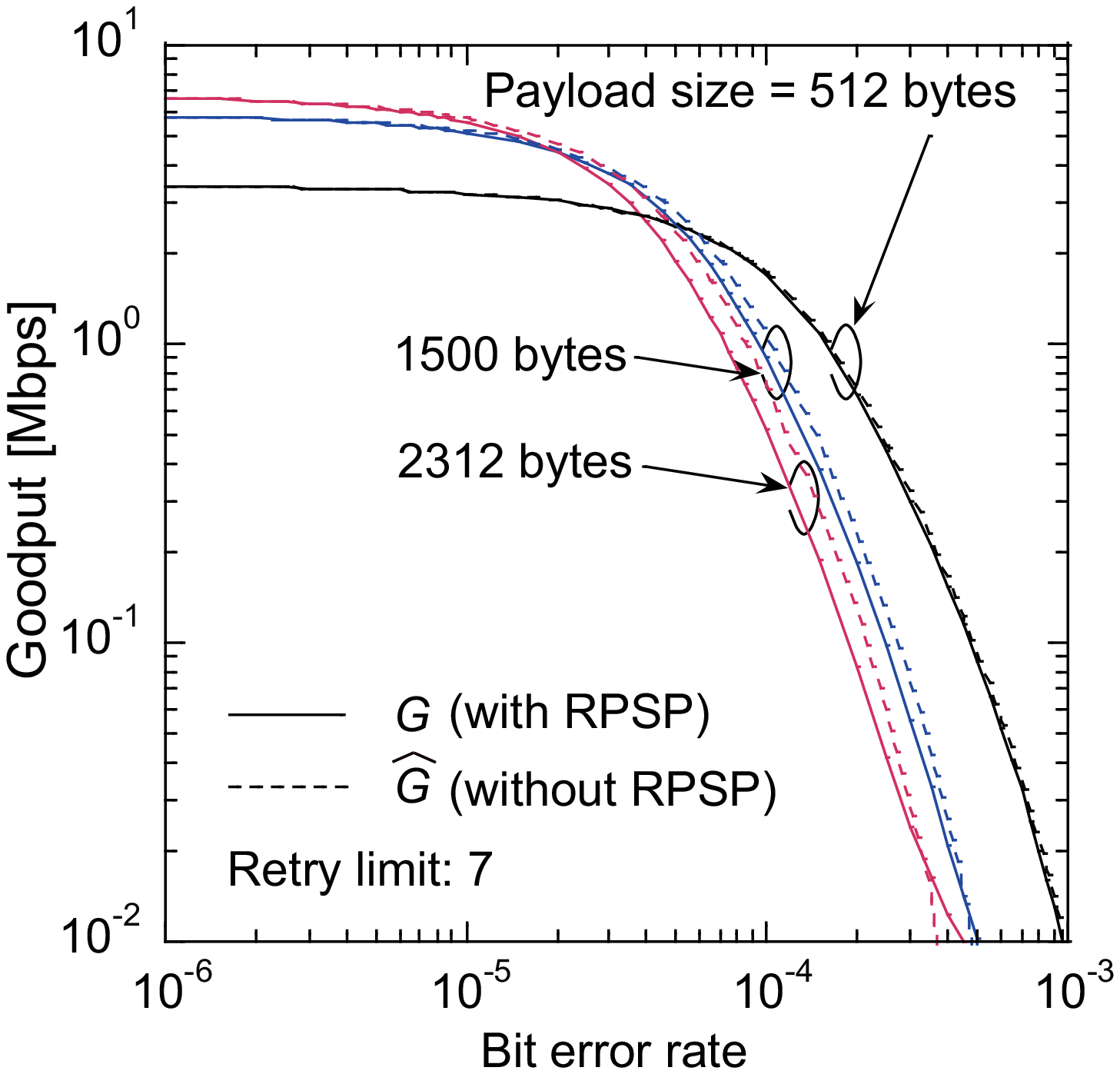, width=6cm} \\
	\footnotesize{a) Case of static Web objects.} 
\end{minipage}
\hfill
\begin{minipage}[t]{8cm}
\centering
	\epsfig{file=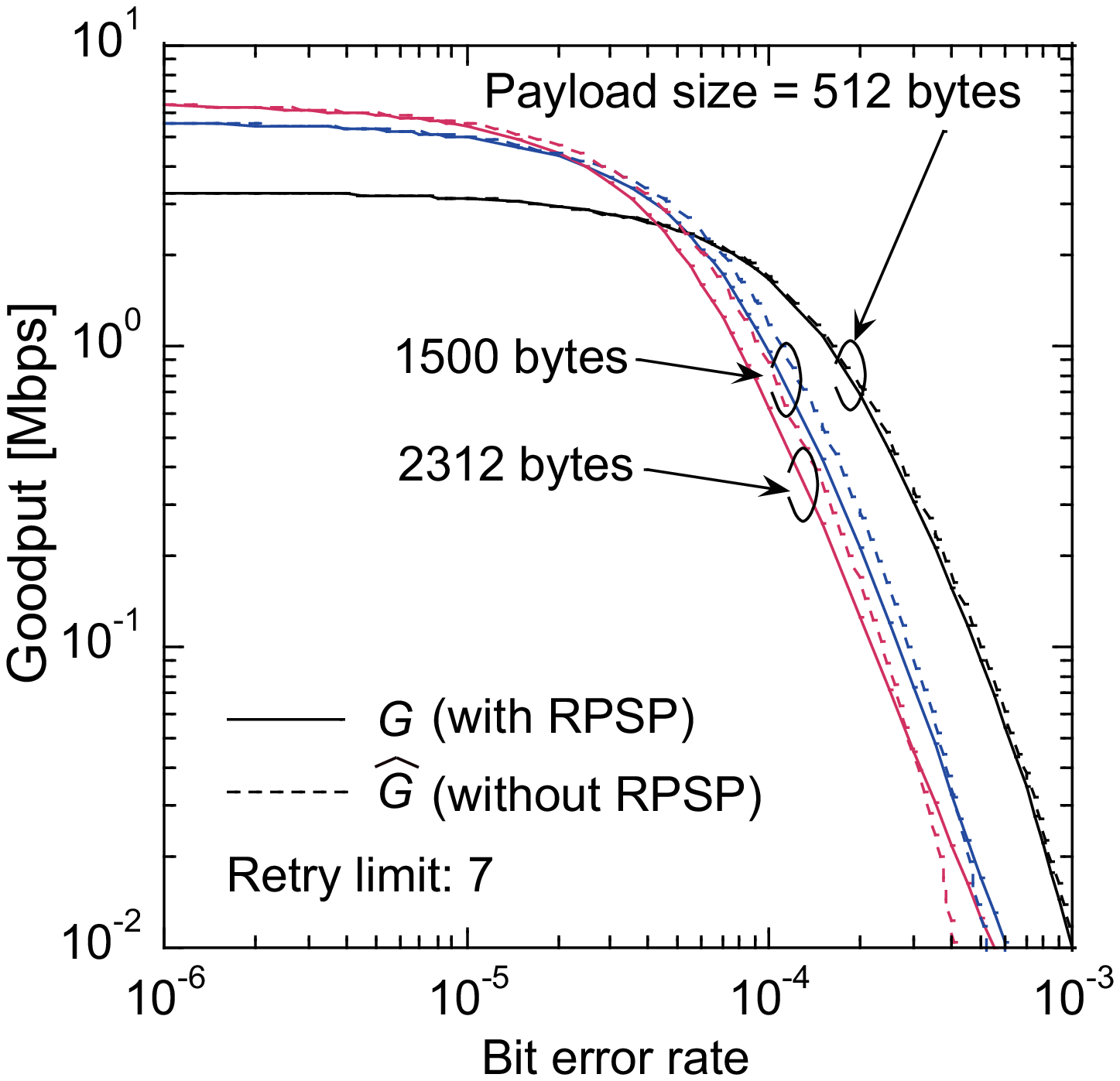, width=6cm} \\
	\footnotesize{b) Case of dynamic Web objects.}
\end{minipage}\\
\caption{Goodput $G$ and $\hat{G}$ versus bit error rate $p_e$ for different payload sizes $\ell_d$
	when limit retry $n_\textrm{RL}$ is $7$.}
\label{fig: goodput}
\end{figure}

In this subsection,
we investigate the RPSP effect on goodput.
To do this,
we introduce $\hat{G}$ which is obtained from the approximation
of $F^{(q)}(x) = \textbf{1}(x - \ell^{(p)})$.
Thus,
\begin{align}
\label{eq: G hat}
	\hat{G} &= \cfrac{
		\Pr.(\text{\it delivery} \, | \, L^{(p)} = \ell^{(p)}) \, 
		(\ell^{(p)} - \ell_h^{({\rm R})})
		}
		{
		E\left[R + 1 \, | \,  L^{(p)} = \ell^{(p)} \right]
			\,E\left[T^{(cycle)} \, | \,  L^{(p)} =\ell^{(p)} \right]
		}.
\end{align}
Cleary, the value of $\hat{G}$ is equal to that of $\hat{G}$
when no transferred packet loss happens
because RPSP effect disappears (see Example~\ref{ex: no frame loss}).

Figures~\ref{fig: goodput} (a) and (b) show
$G$ and $\hat{G}$ versus bit error rate $p_e$ for different payload sizes $\ell_d$
when limit retry $n_\textrm{RL}$ is $7$
in the cases of static and dynamic Web objects, respectively.
From these figures,
we find that RPSP leads to overestimate goodput obtained from the traditional model
which assume that the transferred packets is constant in size.
As similar to the results mentioned in the preceding subsection,
we find that the RPSP effect on goodput appears
when the bit error rate $p_e$ exceeds $10^{-5}$ and payload size $\ell_d$ exceeds $1500$~bytes.

	\begin{figure}
	\centering
	\epsfig{file=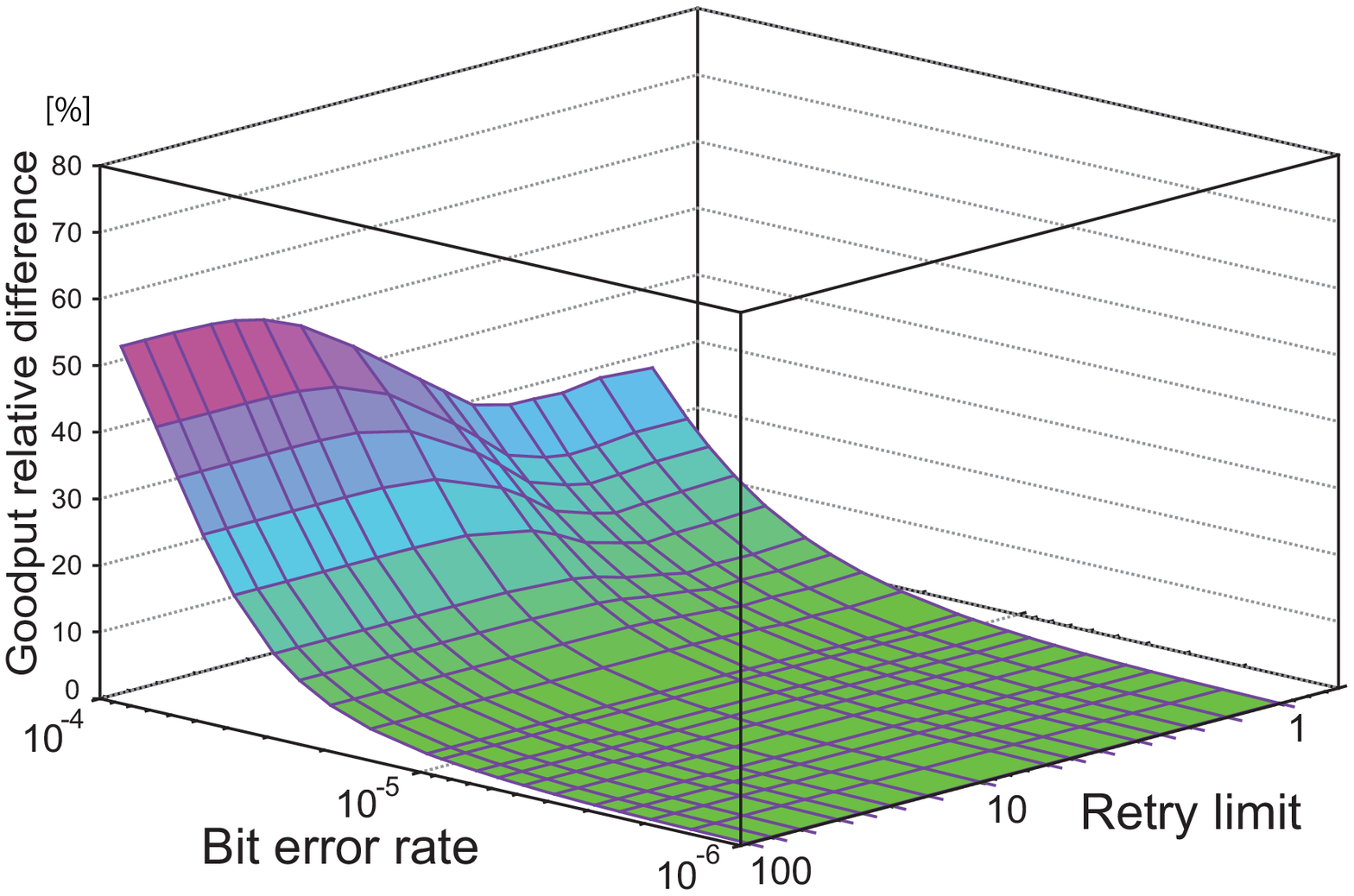, width=9cm} \\
	\caption{Goodput relative difference $(\hat{G} - G) / G$
	versus bit error rate $p_e$ and retry limit $n_\textrm{RL}$
	when payload size $\ell_d$ is 2312~bytes
	in the case of static Web objects.}
	\label{fig: 3dg}
	\end{figure}

Figure~\ref{fig: 3dg} shows goodput relative difference $(\hat{G} - G) / G$
versus bit error rate $p_e$ and retry limit $n_\textrm{RL}$
when payload size $\ell_d$ is 2312~bytes
in the case of static Web objects.
From this figure,
we find that the effect of RPSP on goodput appears stronger when bit error rate $p_e$ increases
for large enough retry limits.

\section{Conclusion}
\label{sec: conclusion}

In this paper,
we have described a data-unit-size distribution model to represent
the retransmitted packet size (RPSP) property and message segmentation behavior
when frames are independently lost and
they are recovered by a stop-and-wait protocol.
RPSP means that all transferred packets at retransmissions
with the same sequence number
have the same size at the original transmission,
which is identical to the packet generated from a message,
namely, generated packets.
Moreover, we have derived the goodput formula 
using an approach to derive the data-unit-size distribution.
We have shown that the RPSP effect appears stronger 
when the maximum generated packet size is larger than the mean generated packet size.
From numerical results,
we have demonstrated that 
the RPSP effect on frame size distributions and goodput appears
when the bit error rate exceeds $10^{-5}$ and payload size exceeds $1500$~bytes
in a scenario where static Web objects are delivered over an IEEE 802.11 DCF wireless network.

The remaining issues include modeling a scenario where
the collisions happen over a wireless network with bit errors occurring in burst.

\section*{Acknowledgment}

This work was supported by JSPS KAKENHI Grant Number JP15K00139. 

%%%%%%% 付録 %%%%%%%%%%%%%%%%%%%%%

\appendices

\section{Proof of Proposition~\ref{pro: F q}}
\label{ap: F q}

First, without loss of generality,
we consider the case of discrete message size distributions,
resulting in the form of discrete generated packet size distributions
given by \eqref{eq: F p d another 1}.
Substituting \eqref{eq: F p d another 1} into \eqref{eq: F q},
we have 
\begin{align}
\label{eq: ap F q d}
	F^{(q)}(x) = \sum_{i=0}^{n_d} w^{(q)}_i \textbf{1}(x - \ell^{(p)}_i),
\end{align}
where $w^{(q)}_i$ for $i=0, 1, \cdots, n_d$ is given by 
\begin{align}
\label{eq: ap w q d}
	w^{(q)}_i &= 
		\cfrac{
		w^{(p)}_i \, E\left[R + 1 \, | \,  L^{(p)} = \ell^{(p)}_i \right]
		}
		{
		E\left[R + 1\right]
		} \notag \\
		&= 
		\cfrac{
		w^{(p)}_i \, E\left[R + 1 \, | \,  L^{(p)} = \ell^{(p)}_i \right]
		}
		{
		\displaystyle\sum_{j=0}^{n_d} \, w^{(p)}_j \,
			E\left[R + 1 \, | \,  L^{(p)} = \ell^{(p)}_j \right]
		}.	
\end{align}
To derive \eqref{eq: ap F q d} and \eqref{eq: ap w q d},
we introduce the following notations
of the generated packet of size equal to 
$\ell^{(p)}_i$ for $i=0, 1, \cdots, n_d$:

\begin{description}

    \item[$M_i (t)$: ]\rule{0in}{0in}
		number of attempts of transmissions
		of transferred packets
		prior to time $t$,

    \item[$Q_{i, \kappa}(t)$: ]
		\rule{0in}{0in}
		number of attempts of transmissions
			of transferred packets
			that are created from the generated packet
			with {\tt seqNum} of $\kappa$
			prior to time $t$.

\end{description}

A sender transmits the transferred packet
of which {\tt seqNum} is $\kappa$
and size is $\ell^{(p)}_i$
$Q_{i, \kappa}(t)$ times prior to time $t$.
From the argument of a probability mass function,
the form of $w^{(q)}_i$ can be written as
\begin{align}
\label{eq: ap w q}
	&w^{(q)}_i = 
	\lim_{t \to \infty}
	\cfrac{
		\displaystyle\sum_{\kappa=1}^{M_i (t)} \, Q_{i, \kappa}(t)
	}
	{
		\displaystyle\sum_{i=0}^{n_d} \sum_{\kappa=1}^{M_i(t)} \, Q_{i, \kappa}(t)
	}
	=
	\left(
	\lim_{t \to \infty}
	\cfrac{
		\displaystyle\sum_{j=0}^{n_d} M_j(t)
	}
	{
		\displaystyle\sum_{i=0}^{n_d} \sum_{\kappa=1}^{M_i(t)} \, Q_{i, \kappa}(t)
	}
	\right)
	\left(
	\lim_{t \to \infty}
	\cfrac{M_i (t)}
	{
	\displaystyle\sum_{j=0}^{n_d} M_j(t)
	}
	\right)
	\left(
	\lim_{t \to \infty}
	\cfrac{
		\displaystyle\sum_{\kappa=1}^{M_i (t)} \, Q_{i, \kappa}(t)
	}		
	{M_i (t)}
	\right).
\end{align}
The form of $w^{(p)}_i$ in \eqref{eq: ap w q d} is given by
\begin{align}
\label{eq: ap w p}
	w^{(p)}_i &=
	\Pr.(\text{the generated packet of size is equal to 
		$\ell^{(p)}_i$}) 
	=\lim_{t \to \infty}
	\cfrac{M_i (t)}
	{
	\displaystyle\sum_{j=0}^{n_d} M_j(t)
	}.
\end{align}

Let $R_\kappa$ be the number of retransmissions
of the generated packet of which {\tt seqNum} is $\kappa$.
Under assumption \texttt{A2},
$\{R_\kappa\}$ forms a sequence of mutually independent and identically
distributed random variables with finite value of
$E[R_\kappa](\stackrel{\triangle}{=}E[R])$.
From the Law of Large Numbers,
we have
\begin{align}
\label{eq: ap R+1}
	E[R + 1] &= 	
	\lim_{t \to \infty}
	\cfrac{
		\displaystyle\sum_{i=0}^{n_d} \sum_{\kappa=1}^{M_i(t)} \, Q_{i, \kappa}(t)
	}
	{
		\displaystyle\sum_{j=0}^{n_d} M_j(t)		
	}, \\
\intertext{and}
\label{eq: ap R+1 L}
	E[R + 1 \, | \, L^{(p)} = \ell^{(p)}_i] &= 
	\lim_{t \to \infty}
	\cfrac{
		\displaystyle\sum_{\kappa=1}^{M_i (t)} \, Q_{i, \kappa}(t)
	}		
	{M_i (t)}.		
\end{align}
Substituting \eqref{eq: ap w p}, \eqref{eq: ap R+1} and \eqref{eq: ap R+1 L}
into \eqref{eq: ap w q},
we obtain \eqref{eq: ap F q d} and \eqref{eq: ap w q d}.

Next, we provide an alternative derivation of \eqref{eq: F q}.
Consider a packet size sequence
$\{L^{(q)}_n; n \in \mathcal{N}(\stackrel{\triangle}{=}\{1, 2, \cdots\})\}$
where $L^{(q)}_n$ means the transferred packet size
of the $n$th transmission.
Forming 
transferred packets with the same {\tt seqNum}
a group,
we constitute a sequence $\{L^{(q)}_n\}$ expressed as
\begin{multline}
\{L^{(q)}_n; n \in \mathcal{N}\} = \\
	\Big\{
	\overbrace{
		L^{(p)}_1,\cdots, L^{(p)}_1, L^{(p)}_1 
		}^{R_1 + 1}, \,
	\overbrace{
		L^{(p)}_2,\cdots, L^{(p)}_2, L^{(p)}_2
		}^{R_2 + 1}, 
	\cdots, 
	\overbrace{
		L^{(p)}_\kappa,\cdots, L^{(p)}_\kappa, L^{(p)}_\kappa
		}^{R_\kappa + 1}, 
	\cdots
	\Big\}.
\label{eq: L q}
\end{multline}
As shown in \eqref{eq: L q},
the random variable $L^{(p)}_\kappa$ appears $R_\kappa + 1$
times consecutively in the sequence of the transferred packets
with {\tt seqNum} of $\kappa$.
Therefore, we obtain \eqref{eq: F q}.

\section{Proof of Proposition~\ref{pro: G}}
\label{ap: goodput}

	\begin{figure}
	\centering
	\epsfig{file=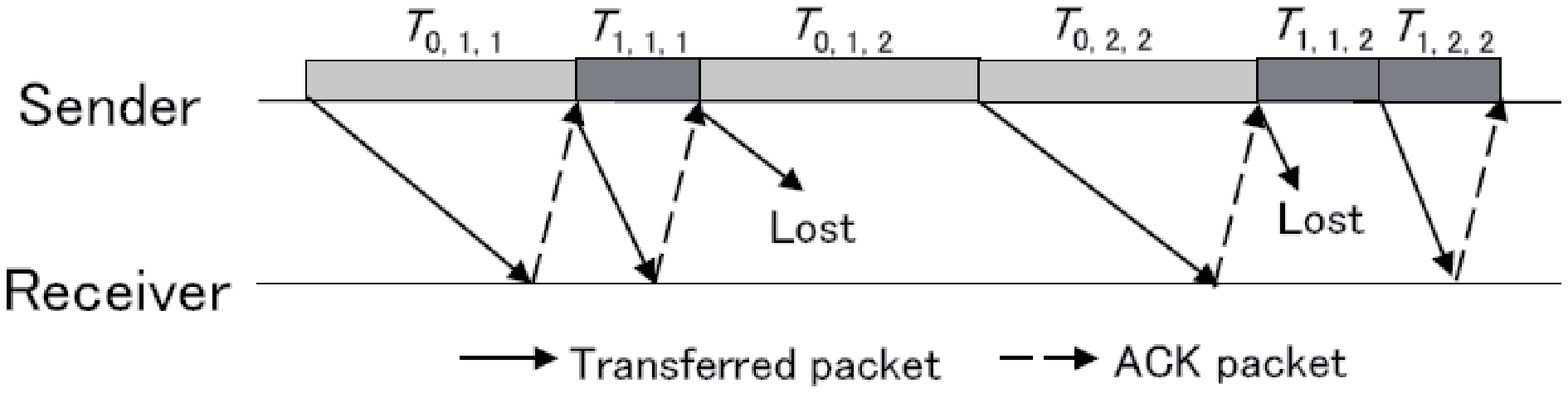, width=10cm} \\
	\caption{The example of $\{T_{i, j, \kappa}\}$ under a heavy traffic condition
	in the case of $n_d$ equal to one.}
	\label{fig: t}
	\end{figure}

Similar to the proof mentioned in Appendix~\ref{ap: F q},
we consider the case of discrete message size distributions
given by \eqref{eq: F p d another 1}.
Substituting \eqref{eq: F p d another 1} into \eqref{eq: G},
we have 
\begin{align}
\label{eq: G d}
	G &= \cfrac{
		\displaystyle\sum_{i=0}^{n_d} \, w^{(p)}_i \, 
			\Pr.(\text{\it delivery} \, | \, L^{(p)} = \ell^{(p)}_i) 
				\, \left(\ell^{(p)}_i - \ell_h^{({\rm R})}\right)}
		{
		\displaystyle\sum_{i=0}^{n_d} \, w^{(p)}_i \, 
			E\left[R + 1 \, | \,  L^{(p)} = \ell^{(p)}_i \right]
			\,E\left[T^{(cycle)} \, | \,  L^{(p)} = \ell^{(p)}_i \right]
		}.
\end{align}

To derive \eqref{eq: G d},
we introduce the following additional notations
for the generated packet of size equal to 
$\ell^{(p)}_i$ for $i=0, 1, \cdots, n_d$:

\begin{description}

    \item[$N_i (t)$: ]\rule{0in}{0in}
		number of successful transmissions
		of transferred packets
		prior to time $t$,

    \item[$T_{i, j, \kappa}$: ]\rule{0in}{0in}	
	transmission of the $j(\le Q_{i, \kappa}(t))$th attempt
	for the transferred packet of which {\tt seqNum} is $\kappa(\le M_i (t))$.
	The example of $\{T_{i, j, \kappa}\}$ under a heavy traffic condition
	in the case of $n_d$ equal to one is shown in Fig.~\ref{fig: t}.

\end{description}

For large enough $t$, we have
\begin{align}
\label{eq: t RT}
    t &\approx \sum_{i=0}^{n_d} \,\sum_{\kappa=1}^{M_i (t)} 
		\,\sum_{j=1}^{Q_{i, \kappa}(t)}  \, T_{i, j, \kappa}.
\end{align}
The definition of goodput yields
\begin{align}
\label{eq: G another}
    G &= \lim_{t \to \infty}
	\cfrac{
		\displaystyle  \sum_{i=0}^{n_d} N_i(t) \, \left(\ell^{(p)}_i - \ell_h^{({\rm R})}\right)
		}
		{
		t
		}.
\end{align}
Substituting \eqref{eq: t RT} into  \eqref{eq: G another},
we have
\begin{align}
\label{eq: G 1}
    G &= \lim_{t \to \infty}
	\cfrac{
		\displaystyle  \sum_{i=0}^{n_d} N_i(t) \, \left(\ell^{(p)}_i - \ell_h^{({\rm R})}\right)
		}
		{
		\displaystyle
		\sum_{i=0}^{n_d} \, \sum_{j=1}^{Q_{i, \kappa}(t)} \, 
		\sum_{\kappa=1}^{M_\kappa (t)} \, T_{i, j, \kappa}	
		}
	= 
		\left(
		\lim_{t \to \infty} 
	    	\cfrac{
		\displaystyle\sum_{i=0}^{n_d} N_i(t) \, \left(\ell^{(p)}_i - \ell_h^{({\rm R})}\right)
		}
		{
		\displaystyle\sum_{i=0}^{n_d} \, \sum_{\kappa=1}^{M_i(t)} \, Q_{i, \kappa}(t) 
		}
		\right)
		\left(		
	   	\lim_{t \to \infty}
		\cfrac{
		\displaystyle\sum_{i=0}^{n_d} \sum_{\kappa=1}^{M_\kappa(t)} \, 
		Q_{i, \kappa}(t)
		}
		{
		\displaystyle
		\sum_{i=0}^{n_d} \, \sum_{j=1}^{Q_{i, \kappa}(t)} \, 
		\sum_{\kappa=1}^{M_i (t)}
		 \, T_{i, j, \kappa}
		}
		\right).
\end{align}

The form of $\Pr.(\text{\it delivery} \, | \, L^{(p)} = \ell^{(p)}_i)$
is given by 
\begin{align}
\label{eq: ap delivery}
	\Pr.(\text{\it delivery} \, | \, L^{(p)} = \ell^{(p)}_i) 
	&= \lim_{t \to \infty} \cfrac{N_i(t)}{M_i(t)}.
\end{align}
From \eqref{eq: ap w p}, \eqref{eq: ap R+1} and \eqref{eq: ap delivery},
we have
\begin{align}
	\lim_{t \to \infty} 
	    \cfrac{N_i(t)}
		{\displaystyle\sum_{i=0}^{n_d} \, \sum_{\kappa=1}^{M_i(t)}\,
		Q_{i, \kappa}(t) 
		} 
	&=
	\left(
	\lim_{t \to \infty}
	\cfrac{
		\displaystyle\sum_{j=0}^{n_d} \, M_j(t)
		}
	{\displaystyle\sum_{i=0}^{n_d} \, \sum_{\kappa=1}^{M_i(t)} \,  Q_{i, \kappa}(t)
	} 
	\right)
	\left(
	\lim_{t \to \infty}\cfrac{M_i(t)}
		{\displaystyle\sum_{j=0}^{n_d} \, M_j(t)}
	\right)	
	\left(
	\lim_{t \to \infty}\cfrac{N_i(t)}{M_i(t)}
	\right)	 \notag\\
	&= 
	\cfrac{1}{E[R+1]} \cdot 
		w^{(p)}_i \cdot \Pr.(\text{\it delivery} \, | \, L^{(p)} = \ell^{(p)}_i).
\end{align}

The first term of \eqref{eq: G 1} can be rewritten as
\begin{align}
\label{eq: G first term}
	\lim_{t \to \infty} 
	    \cfrac{
		\displaystyle\sum_{i=0}^{n_d} N_i(t) \, \left(\ell^{(p)}_i - \ell_h^{({\rm R})}\right)
		}
		{\displaystyle\sum_{i=0}^{n_d} \sum_{\kappa=1}^{M_i(t)} \, Q_{i, \kappa}(t)}
	&=
	\cfrac{
	\displaystyle\sum_{i=0}^{n_d}
	w^{(p)}_i \, \Pr.(\text{\it delivery} \, | \, L^{(p)} = \ell^{(p)}_i) \, 
	\left(\ell^{(p)}_i - \ell_h^{({\rm R})}\right)
	}
	{E[R+1]}.
\end{align}

Under the assumption of \texttt{A2},
$\{T_{i, 1, 1}, T_{i, 1, 2}, \cdots, T_{i, 2, 1}, T_{i, 2, 2}, \cdots\}$
forms a sequence of mutually independent and identically distibuted random variables
with a common distribution with mean $E[T^{(cycle)} \, | \, L^{(p)} = \ell^{(p)}_i]$.
From the Law of the Large Numbers, we have
\begin{align}
\label{eq: T cycle}
	\lim_{t \to \infty}
		\cfrac{
		\displaystyle\sum_{\kappa=1}^{M_i (t)} \,
		\sum_{j=1}^{Q_{i, \kappa}(t)} \, T_{i, j, \kappa}
		}
		{
		\displaystyle\sum_{\kappa=1}^{M_i(t)} \, Q_{i, \kappa}(t)
		} 
	&= E\left[T^{(cycle)} \, | \, L^{(p)} = \ell^{(p)}_i\right].
\end{align}

The inverse of the last term of \eqref{eq: G 1} can be rewritten as
\begin{multline}
\label{eq: G last term}
	   \lim_{t \to \infty}
		\cfrac{
		\displaystyle\sum_{i=0}^{n_d} \, \sum_{\kappa=1}^{M_i (t)} \, 
		\sum_{j=1}^{Q_{i, \kappa}(t)} \, T_{i, j, \kappa}
		}
		{
		\displaystyle\sum_{i=0}^{n_d} \sum_{\kappa=1}^{M_i(t)} \, Q_{i, \kappa}(t)
		} \\
	=  \displaystyle\sum_{i=0}^{n_d}
		\left(
		\lim_{t \to \infty}
		\cfrac{
		\displaystyle\sum_{j=0}^{n_d} M_j(t)		
		}
		{
		\displaystyle\sum_{i=0}^{n_d} \sum_{\kappa=1}^{M_i(t)} \, Q_{i, \kappa}(t)
		}
		\right)
		\left(
		\lim_{t \to \infty}\cfrac{M_i(t)}
		{\displaystyle\sum_{j=0}^{n_d} \, M_j(t)}
		\right)	
		\left(		
		\lim_{t \to \infty}
		\cfrac{
		\displaystyle\sum_{\kappa=1}^{M_i (t)} \, Q_{i, \kappa}(t)
		}		
		{M_i (t)}
		\right)	
		\left(	
		\lim_{t \to \infty}
		\cfrac{
			\displaystyle\sum_{\kappa=1}^{M_i (t)} \, 
			\sum_{j=1}^{Q_{i, \kappa}(t)} \, T_{i, j, \kappa}
			}
			{
			\displaystyle\sum_{\kappa=1}^{M_i (t)} \, Q_{i, \kappa}(t)
			}
		\right)
\end{multline}
Substituting \eqref{eq: ap w p}, \eqref{eq: ap R+1}, \eqref{eq: ap R+1 L} and \eqref{eq: T cycle}
into the above equation,
we have
\begin{align}
\label{eq: G last term final}
	   \lim_{t \to \infty}
		\cfrac{
		\displaystyle
		\sum_{i=0}^{n_d} \, \sum_{\kappa=1}^{M_i (t)} \, 
		\sum_{j=1}^{Q_{i, \kappa}(t)} \, T_{i, j, \kappa}
		}
		{
		\displaystyle\sum_{i=0}^{n_d} \sum_{\kappa=1}^{M_i(t)} \, Q_{i, \kappa}(t)
		}
	&=
	\cfrac{
	\displaystyle\sum_{i=0}^{n_d} w^{(p)}_i \, E\left[R+1 \, | \,  L^{(p)} = \ell^{(p)}_i \right] \,
	E\left[T^{(cycle)} \, | \,  L^{(p)} = \ell^{(p)}_i \right]
	}
	{
	E[R+1]
	}.  
\end{align}
Substitution of \eqref{eq: G first term} and \eqref{eq: G last term final} into 
\eqref{eq: G 1} yields \eqref{eq: G d}.

\section{Derivation of \eqref{eq: E T cycle L}}
\label{ap: E T cycle}

Because no collision occurs,
from the argument of \cite{CHA03},
we have
\begin{align}
\label{eq: E T cycle L 1}
    E\left[T^{(cycle)} \, | \, L^{(p)} = x \right] &= 
	\cfrac{\left(1-\tau(x)\right) \, \sigma}{\tau(x)}
	+ \left(1 - g(x)\right) \, t_{\textrm{suc}}(x) 
	+ g(x) \, t_{\textrm{bit}}(x),	
\end{align}
where
$\tau(x)$ is the probability that a sender can transmit a transferred packet
of size equal to $x$.
From the argument of \cite{BIA05},
$\tau(x)$ is given by
\begin{align}
\label{eq: tau ell} 
    \tau(x) 
	&=\cfrac{1}
	{
	1+\cfrac{1-g(x)}
		{1-\left\{g(x)\right\}^{n_\text{RL}+1}
	}
	\displaystyle\sum_{r=0}^{n_\text{RL}}
	b_r \, \left\{g(x)\right\}^r 
	}.
\end{align}
Substitution \eqref{eq: tau ell} into \eqref{eq: E T cycle L 1},
we obtain \eqref{eq: E T cycle L}.

\section{Proof of Conjecture~\ref{conjecture: ell q max}}
\label{ap: w q infty}

Suppose that the generated packets sizes follow
the discrete distribution
given by \eqref{eq: F p d another 1}.
By substitution of \eqref{eq: F p d another 1} into \eqref{eq: F q}, 
the transferred packet size distribution $F^{(q)}(\cdot)$ is given by
\begin{align}
\label{eq: F q d}
	F^{(q)}(x) &\stackrel{\triangle}{=}  
	\sum_{i=0}^{n_d} w^{(q)}_i \textbf{1}(x - \ell^{(q)}_i),
\end{align}
where
\begin{align}
	w^{(q)}_i &=
	\cfrac{
		w^{(p)}_i \, h(\ell^{(p)}_i, n_\textrm{RL})
		}
		{
		\displaystyle\sum_{j=0}^{n_d} \, w^{(p)}_j \, h(\ell^{(p)}_j, n_\textrm{RL})
		},
	\quad i = 0, 1, \cdots, n_d,		
\end{align}
because $E[R + 1 \, | \,  L^{(p)} = \ell^{(p)}_i] = h(\ell^{(p)}_i, n_\textrm{RL})$
and
$E[R + 1] = \sum_{j=0}^{n_d} \, w^{(p)}_j \, h(\ell^{(p)}_j, n_\textrm{RL})$.

Let $i_{\max}$ be the index corresponding to 
the maximum generated packet size $\ell^{(p)}_{\max}$.
Thus,
\begin{align}
i_{\max} = \argmax_{i \in \{0,1,\cdots, n_d\}}\{\ell^{(p)}_i\}.
\end{align}

We let $\bar{w}^{(q)}_i$  
be a finite limit of the weight corresponding to
discrete transferred packet size $\ell^{(p)}_i$
as $p_e \to 1$ and $n_\textrm{RL} \to \infty$
for $i=0, 1, \cdots, n_d$.
From 
$\lim_{n_\textrm{RL} \to \infty} h(x, n_\textrm{RL}) = 1 / (1 - g(x)) 
	= 1 / (1 - p_e)^{x + \ell_h^{({\textrm{L}})}}$
	if $0 \le g(x) < 1$,
we have 
\begin{align}
\label{eq: n infty p_e 1 h}
\bar{w}^{(q)}_i &= \lim_{p_e \to 1} \, \lim_{n_\textrm{RL} \to \infty}
		\cfrac{
		w^{(p)}_i \, h(\ell^{(p)}_i, n_\textrm{RL})
		}
		{
		\displaystyle\sum_{j=0}^{n_d} \, w^{(p)}_j \, h(\ell^{(p)}_j, n_\textrm{RL})
		} 
	= \lim_{p_e \to 1} \, 
	\cfrac{w^{(p)}_i}{(1 - p_e)^{\ell^{(p)}_i + \ell_h^{({\textrm{L}})}}}
	\cfrac{1}
		{
		\displaystyle\sum_{j=0}^{n_d}
	    	\cfrac{w^{(p)}_j}{(1 - p_e)^{\ell^{(p)}_j + \ell_h^{({\textrm{L}})}}}
		} \notag \\
	&= \lim_{p_e \to 1} \,
	\cfrac{w^{(p)}_i}{(1 - p_e)^{\ell^{(p)}_i + \ell_h^{({\textrm{L}})}}}
	\cfrac{(1 - p_e)^{\ell^{(p)}_{\max} + \ell_h^{({\textrm{L}})}}}
		{
		\displaystyle\sum_{j=0}^{n_d}
	    	\cfrac{w^{(p)}_j}{(1 - p_e)^{\ell^{(p)}_{\max} - \ell^{(p)}_j}}
		} \notag \\
	&=\lim_{p_e \to 1} \,
	\cfrac{w^{(p)}_i}{(1 - p_e)^{\ell^{(p)}_i + \ell_h^{({\textrm{L}})}}}
	\cfrac{(1 - p_e)^{\ell^{(p)}_{\max} + \ell_h^{({\textrm{L}})}}}
	{
		\displaystyle \sum_{{j=0, j \ne i_{\max}}}^{n_d}w^{(p)}_j
		(1 - p_e)^{\ell^{(p)}_{\max} - \ell^{(p)}_j}
		+ w^{(p)}_{\max}
	}  \notag \\
	&=
	\begin{cases}
	1,	& \text{for $i = i_{\max}$} \\
	0,	& \text{for $i \ne i_{\max}$}	
	\end{cases}.
\end{align}
Thus,
we have 
\begin{align}
F^{(q)}(x) \to \textbf{1}(x - \ell^{(p)}_{\max}),
	\qquad\text{as $p_e \to 1$.}
\end{align}
Therefore, we obtain \eqref{eq: ell q max}.

\small

\bibliographystyle{IEEEtran}
\bibliography{paper}

\end{document}